%% file: sortedsum.tex
\documentclass{article}
\usepackage{graphicx,epsfig,color}

\usepackage{enumerate}

\input{mac90}

\squeeze

\newcommand{\tightcomplexityZ}{
$O({1\over \epsilon}\min(\log n, {\log ({x_{max}\over
x_{min}})})\cdot (\log {1\over \epsilon}+\log\log n))$}

\begin{document}%\baselineskip 20pt

\date{}

\title{On the Complexity of Approximate Sum of Sorted List~\thanks{This research is supported in part by the NSF Early Career Award CCF 0845376. The first version is on December 2,
2011, and it is revised on 12/18/2011, 1/16/2012, and 1/21/2012.}}

\author{
Bin Fu
\\ \\
Department of Computer Science\\
 University of Texas-Pan American\\
 Edinburg, TX 78539, USA\\
% Email: binfu@utpa.edu\\
 Email: binfu@cs.panam.edu\\
} \maketitle

%\centerline {First Version on December 2, 2011}

%\centerline {Revised Version on December 18, 2011}

\begin{abstract} We consider the complexity for computing
the approximate sum $a_1+a_2+\cdots+a_n$ of a sorted list of numbers
$a_1\le a_2\le \cdots\le a_n$. We show an algorithm that computes an
$(1+\epsilon)$-approximation for the sum of a sorted list of
nonnegative numbers in an
%$O({1\over \epsilon}(\log n)\cdot (\log {1\over \epsilon}+\log\log n))$
%$O({1\over \epsilon}(\log {1\over \epsilon})\cdot (\log{1\over \epsilon}+\log n))$
~\tightcomplexityZ~time, where $x_{max}$ and $x_{min}$ are the
largest and the least positive elements of the input list,
respectively.
% We first derive an $O(\log \log n)$ time approximation algorithm that finds an approximate
%region of the array for holding the items of size at least a
%threshold $b$. Our approximate sum algorithm is derived with it as a
%submodule. We also show an $\Omega(\log \log n)$ lower bound for
%approximate region algorithms for the sum of a sorted list with only
%nonnegative elements.
We  prove a lower bound $\Omega(\min(\log n,\log ({x_{max}\over
x_{min}}))$ time for every $O(1)$-approximation algorithm for the
sum of a sorted list of nonnegative elements.
%The lower bound almost matches the upper bound.
We also show that there is no sublinear time approximation algorithm
for the sum of a sorted list that contains at least one negative
number.
\end{abstract}

\section{Introduction}

Computing the sum of a list of numbers is a classical problem that
is often found inside the high school textbooks. There is a famous
story about Karl Friedrich Gauss who computed $1+2+\cdots +100$ via
rearranging these terms into $(1+100)+(2+99)+...+(50+51)=50\times
101$, when he was seven years old, attending elementary school. Such
a method is considered an efficient algorithm for computing a class
of lists of increasing numbers. Computing the sum of a list of
elements has many applications, and is ubiquitous in software
design. In the classical mathematics, many functions can be
approximated by the sum of simple functions via Taylor expansion.
This kind of approximation theories is in the core area of
mathematical analysis. In this article we consider if there is an
efficient way to compute the sum of a general list of nonnegative
numbers with nondecreasing order.
%For example, it may be time consuming
%for computing $\sum_{m=1}^n {1\over m\log m}$ if $n$ is large.
%Some important  functions can be approximated by the sum of a list
%of monotonic real numbers from their Taylor expansions.
%For example, $e^x=\sum_{i=1}^{\infty}{x^i\over i!}$.
%Some other functions can be approximated by  $\sum_{i=1}^{\infty}
%a_i-\sum_{i=1}^{\infty} b_i$, where both $\sum_{i=1}^{\infty} a_i$
%and $\sum_{i=1}^{\infty} b_i$ are the sums of monotonic nonnegative
%real numbers.
%For example, $sin(x)=x-{x\over 3!}+\cdots+(-1)^i{x^{2i+1}\over
%(2i+1)!}+\cdots=(x+{x^5\over 5!}+\cdots)-({x^3\over 3!}+{x^7\over 7!}+\cdots)$.

 Let
$\epsilon$ be a real number at least $0$. Real number $s$ is an
$(1+\epsilon)$-approximation for the sum problem $a_1,a_2,\cdots,
a_n$ if ${\sum_{i=1}^n a_i\over 1+\epsilon}\le s\le
(1+\epsilon)\sum_{i=1}^na_i$. Approximate sum problem was studied in
the randomized computation model. Every $O(1)$-approximation
algorithm with uniform random sampling requires $\Omega(n)$ time in
the worst case if the list of numbers in $[0,1]$ is not sorted.
Using $O({1\over \epsilon^2}\log {1\over \delta})$ random samples,
one can compute the $(1+\epsilon)$-approximation for the mean, or
decide if it is at most $\delta$ for a list numbers in
$[0,1]$~\cite{Hoefding63}. Canetti, Even, and
Goldreich~\cite{CanettiEvenGoldreich95} showed that the sample size
is tight.  Motwani, Panigrahy, and Xu~\cite{MotwaniPanigrahyXu07}
showed an $O(\sqrt{n})$ time approximation scheme for computing the
sum of $n$ nonnegative elements. There is a long history of research
for the accuracy of summation of floating point numbers
 (for examples, see~\cite{Kahan65,Bresenham65,Anderson99,DemmelHida03,Espelid95,Gregory72,Higham93,Knuth98,Linz70,Malcolm71,Priest92,ZhuYongZheng05}).
The efforts were mainly spent on finding algorithms with small
rounding errors.
%We have not found any existing article about the
%approximate algorithm for computing the sum of a sorted list.

We investigate the complexity for computing the approximate sum of a
sorted list. When we have a large number of data items and need to
compute the sum, an efficient approximation algorithm becomes
important. Par-Heled developed an coreset approach for a more
general problem. The method used in his paper implies an $O({\log
n\over \epsilon})$ time approximation algorithm for the approximate
sum of sorted nonnegative numbers~\cite{Har-Peled06}. The coreset is
a subset of numbers selected from a sorted input list, and their
positions only depends on the size $n$ of the list, and independent
of the numbers. The coreset of a list of $n$ sorted nonnegative
numbers has a size $\Omega(\log n)$. This requires the algorithm
time to be also $\Omega(\log n)$ under all cases.

 We show an algorithm that gives an
$(1+\epsilon)$-approximation for the sum of a list of sorted
nonnegative elements in
%$O({1\over \epsilon}(\log n)\cdot (\log {1\over \epsilon}+\log\log n))$
%$O({1\over \epsilon}(\log {1\over \epsilon})\cdot (\log{1\over \epsilon}+\log n))$
~\tightcomplexityZ~time,  where $x_{max}$ and $x_{min}$ are the
largest and the least positive elements of the input list,
respectively. This algorithm has a comparable complexity with
Par-Heled's algorithm. Our algorithm is of sub-logarithm complexity
when ${x_{max}\over x_{min}}\le n^{1\over (\log\log n)^{1+a}}$ for
any fixed $a>0$. The algorithm is based on a different method, which
is a quadratic region search algorithm, from the coreset
construction used in~\cite{Har-Peled06}.

We also prove a lower bound $\Omega(\min(\log n,\log ({x_{max}\over
x_{min}}))$ for this problem.
%The lower bound almost matches the upper bound.
We first derive an $O(\log \log n)$ time
approximation algorithm that finds an approximate region of the list
for holding the items of size at least a threshold $b$. Our
approximate sum algorithm is derived with it as a submodule. We also
show an $\Omega(\log \log n)$ lower bound for approximate region
algorithms for the sum of a sorted list with only nonnegative
elements.

%A list $A$ of $n$ elements is $k$-sorted if $A[i]\le A[i+k]$ for
%each $i\le n-k$. If the input list is a $k$-sorted nonnegative
%elements (they may be larger than $1$), it computes an
%$(1+\epsilon)$-approximation for the sum problem in time $O({1\over
%\epsilon}(k+\log {1\over \epsilon}+\log\log n)\log n)$.

In Section~\ref{alg-sec}, we present an algorithm that computes
$(1+\epsilon)$-approximation for the sum of a sorted list of
nonnegative numbers in
%$O({1\over \epsilon}(\log n)\cdot (\log {1\over \epsilon}+\log\log n))$ time.
$O({1\over \epsilon}\min(\log n, {\log ({x_{max}\over
x_{min}})})\cdot (\log {1\over \epsilon}+\log\log n))$ time, where
$x_{max}$ and $x_{min}$ are the largest and the least positive
elements of the input list, respectively.  In
Section~\ref{lower-bounds-sec}, we present lower bounds related to
the sum of sorted list.
 In Section~\ref{implementation-sec}, we show the
experimental results for the implementation of our algorithm in
Section~\ref{alg-sec}. This paper contains self-contained proofs for
all its results.
%\section{Notations}

%\begin{definition}

%\end{definition}

\section{Algorithm for Approximate Sum of Sorted List}\label{alg-sec}

In this section, we show a deterministic algorithm for the sorted
elements. We first show an approximation to find an approximate
region of a sorted list with elements of size at least threshold
$b$.

A crucial part of our approximate algorithm for the sum of sorted
list  is to find an approximate region with elements of size at
least a threshold $b$. We develop a method that is much faster than
binary search and it takes $O(\log {1\over \delta}+\log\log n)$ time
to find the approximate region. We first apply the square function
to expand the region and  use the square root function to narrow
down to a region that only has $(1+\delta)$ factor difference with
the exact region. The parameter $\delta$ determines the accuracy of
approximation.
%Such an accuracy is enough to provide the
%corresponding approximation ratio for the approximate algorithm of the sum problem.

\begin{definition}\label{interval-size-def}
 For $i\le j$, let $|[i,j]|$ be the number of integers in the interval $[i,j]$.
\end{definition}

If both $i$ and $j$ are integers with $i\le j$, we have
$|[i,j]|=j-i+1$.

\begin{definition} A list $X$ of $n$ numbers is represented by
an array $X[1,n]$, which has $n$ numbers $X[1],X[2],\cdots, X[n]$.
For integers $i\le j$, let $X[i,j]$ be the sublist that contains
elements $X[i],X[i+1],\cdots, X[j]$. For an interval $R=[i,j]$,
denote $X[R]$ to be $X[i,j]$.
\end{definition}

\begin{definition}
For a sorted list  $X[1,n]$ with nonnegative elements by
nondecreasing order and a threshold $b$, the {\it $b$-region} is an
interval $[n',n]$ such that $X[n',n]$ are the numbers at least $b$
in $X[1,n]$. An {\it $(1+\delta)$-approximation for the $b$-region}
is a region $R=[s,n]$, which contains the last position $n$ of
$X[1,n]$, such that at least ${|R|\over 1+\delta}$ numbers in
$X[s,n]$ are at least $b$, and $[s,n]$ contains all every position
$j$ with $X[j]\ge b$, where $|R|$ is the number of integers $i$ in
$R$.
\end{definition}

\subsection{Approximate Region}\label{appx-region-sec}

The approximation algorithm for finding an approximate $b$-region to
contain the elements at least a threshold $b$ has two loops. The
first loop searches the region by increasing  the parameter $m$ via
the square function. When the region is larger than the exact
region, the second loop is entered. It converges to the approximate
region with a factor that goes down by a square root each cycle.
Using the combination of the square and square root functions makes
our algorithm much faster than the binary search.

In order to simplify the description of the algorithm
Approximate-Region(.), we assume $X[i]=-\infty$ for every $i\le 0$.
It can save the space for the boundary checking when accessing the
list $X$.  The description of the algorithm is mainly based on the
consideration for its proof of correctness. For a real number $a$,
denote $\floor{a}$ to be the largest integer at most $a$, and
$\ceiling{a}$ to be the least integer at least $a$. For examples,
$\floor{3.7}=3$, and $\ceiling{3.7}=4$.

 \vskip 10pt

{\bf Algorithm Approximate-Region($X, b, \delta, n$)}

Input: $X[1,n]$ is a sorted list of $n$ numbers by nondecreasing
order; $n$ is the size of $X[1,n]$; $b$ is a threshold in
$(0,+\infty)$; and $\delta$ is a parameter in $(0,+\infty)$.

%{\obeylines
\begin{enumerate}[1.]
\item\label{first-beg}
\qquad if $(X[n]<b$),  return $\emptyset$;
\item
\qquad if $(X[n-1]<b$), return $[n,n]$;
\item
\qquad if $(X[1]\ge b$), return $[1,n]$;
\item
\qquad let $m_1:=2$;
\item
\qquad while ($X[n-m^2+1]\ge b$) \{
\item
\qquad\qquad let $m:=m^2$;
\item\label{first-end}
\qquad \};
\item\label{second-beg0}
\qquad let $i:=1$;
\item
\qquad let $m_1:=m$;
\item
\qquad let $r_1:=m$;
\item\label{second-beg}
\qquad while ($m_i\ge 1+\delta$) \{
\item
\qquad\qquad let $m_{i+1}:=\sqrt{m_i}$;
\item\label{if-state-app-region}
\qquad\qquad if ($X[n-\floor{m_{i+1}r_i}+1]\ge b)$, then let
$r_{i+1}:=m_{i+1}r_i$;
\item\label{if-else-for-ri+1}
\qquad\qquad else $r_{i+1}:=r_i$;
\item\label{end-second-loop}
\qquad\qquad let $i:=i+1$;
\item\label{second-end}
\qquad \};
\item\label{second-end0}
\qquad return $[n-\floor{m_ir_i}+1,n]$;
\end{enumerate}
{\bf End of Algorithm}
%}

\begin{lemma}\label{sorted-region-lemma} Let $\delta$ be a
parameter in $(0,1)$. Then there is an $O((\log {1\over
\delta})+(\log\log n))$ time algorithm such that given an element
$b$, and a list $A$ of sorted $n$ elements, it finds an
$(1+\delta)$-approximate $b$-region.
\end{lemma}

\begin{proof}
After the first phase (lines~\ref{first-beg} to~\ref{first-end}) of
the algorithm, we obtain number $m$ such that
\begin{eqnarray}
X[n-m+1]&\ge& b, \ \ \ \ \ \mbox{ and}\label{m-boundary-ineqn1}\\
X[n-m^2+1]&<&b.\label{m-boundary-ineqn2}
\end{eqnarray}

As we already assume $X[i]=-\infty$ for every $i\le 0$, there is no
boundary problem for assessing the input list.
 The variable $m$ is an integer in the first phase. Thus,
the boundary point for the region with numbers at least the
threshold $b$ is in $[n-m^2+1,n-m+1]$. The variable $m$ can be
expressed as $2^{2^k}$ for some integer $k\ge 0$ after executing $k$
cycles in the first phase. Thus, the first phase takes $O(\log\log
n)$ time because $m$ is increased to $m^2$ at each cycle of the
first while loop, and $2^{2^k}\ge n$ for $k\ge \log \log n$.

In the second phase (lines~\ref{second-beg0} to~\ref{second-end0})
of the algorithm,  we can prove that $X[n-\floor{r_i}+1]\ge b$ and
$X[n-\floor{m_{i}r_i}+1]<b$ at the end of every cycle (right after
executing the statement at line~\ref{end-second-loop}) of the second
loop (lines~\ref{second-beg} to~\ref{second-end}). Thus, the
boundary point for the region with elements at the threshold $b$ is
in $[n-\floor{m_{i}r_i}+1,n-\floor{r_i}+1]$. The variable $m_i$ is
not an integer after $m_i<2$ in the algorithm. It can be verified
via a simple induction. It is true before entering the second loop
(lines~\ref{second-beg} to~\ref{second-end}) by
inequalities~(\ref{m-boundary-ineqn1}) and
(\ref{m-boundary-ineqn2}). Assume that at the end of cycle $i$,
\begin{eqnarray}
X[n-\floor{r_i}+1]&\ge& b; \ \ \ \mbox{and}\label{hypothesis-1-ineqn}\\
X[n-\floor{m_{i}r_i}+1]&<&b. \label{hypothesis-2-ineqn}
\end{eqnarray}

 Let us consider cycle
$i+1$ at the second loop. Let $m_{i+1}=\sqrt{m_i}$.
\begin{enumerate}
\item
Case 1: $X[n-\floor{m_{i+1}r_i}+1]\ge b$. Let $r_{i+1}=m_{i+1}r_i$
according to line~\ref{if-state-app-region} in the algorithm. Then
$X[n-\floor{r_{i+1}}+1]=X[n-\floor{m_{i+1}r_i}+1]\ge b$. By
inequality (\ref{hypothesis-2-ineqn}) in the hypothesis,
$X[n-\floor{m_{i+1}r_{i+1}}+1]=X[n-\floor{\sqrt{m_i}\sqrt{m_i}r_i}+1]=X[n-\floor{m_ir_i}+1]<b$.
\item
Case 2: $X[n-\floor{m_{i+1}r_i}+1]< b$. Let $r_{i+1}=r_i$ according
to line~\ref{if-else-for-ri+1} the algorithm. We have
$X[n-\floor{r_{i+1}}+1]=X[n-\floor{r_i}+1]\ge b$ by inequality
(\ref{hypothesis-1-ineqn}) in the hypothesis. By inequality
(\ref{hypothesis-2-ineqn}) in the hypothesis,
$X[n-\floor{m_{i+1}r_{i+1}}+1]=X[n-\floor{m_{i+1}r_i}+1]<b$ by the
condition of this case.
\end{enumerate}
Therefore, $X[n-\floor{r_{i+1}}+1]\ge b$ and
$X[n-\floor{m_{i+1}r_{i+1}}+1]<b$ at the end of cycle $i+1$ of the
second while loop.

Every number in $X[n-r_i+1,n]$, which has $r_i$ entries, is at least
$b$, and $X[n-m_ir_i+1,n]$ has $m_ir_i$ entries and $m_i\le
1+\delta$ at the end of the algorithm. Thus, the interval
$[n-m_ir_i+1,n]$ returned by the algorithm is an
$(1+\delta)$-approximation for the $b$-region.

 It takes $O(\log\log n)$ steps for converting
$m$ to be at most $2$, and additional $\log {1\over \delta}$ steps
to make $m$ to be at most $1+\delta$. When $m_i< 1+\delta$, we stop
the loop, and output an $(1+\delta)$-approximation. This step takes
at most $O(\log {1\over \delta}+\log\log n)$ time since $m_i$ is
assigned to $\sqrt{m_i}$ at each cycle of the second loop. This
proves Lemma~\ref{sorted-region-lemma}.
\end{proof}

After the first loop of the algorithm Approximate-Region(.), the
number $m$ is always of the format $2^{2^k}$ for some integer $k$.
In the second loop of the algorithm Approximate-Region(.), the
number $m$ is always of the format $2^{2^k}$ when $m$ is at least
$2$. Computing its square root is to convert $2^{2^k}$ to
$2^{2^{k-1}}$, where $k$ is an integer. Since $(1+{1\over 2^i})\cdot
(1+{1\over 2^i})>(1+{1\over 2^{i-1}})$, we have that $(1+{1\over
2^i})$ is larger than the square root of $(1+{1\over 2^{i-1}})$.  We
may let variable $m_i$ go down by following the sequence
$\{(1+{1\over 2^i})\}_{i=1}^{\infty}$ after $m_i\le 2$. In order
words, let $g(.)$ be an approximate square root function such that
$g(1+{1\over 2^i})=1+{1\over 2^{i+1}}$ for computing the square root
after $m\le 2$ in the algorithm. It has the property $g(m)\cdot
g(m)\ge m$. The assignment $m_{i+1}=\sqrt{m_i}$ can be replaced by
$m_{i+1}=g(m_i)$ in the algorithm. It can simplify the algorithm by
removing the computation of square root while the computational
complexity is of the same order.

\subsection{Approximate Sum}

 We present an algorithm to compute the approximate sum of a list
of sorted nonnegative elements. It calls the module for the
approximate region, which is described in
Section~\ref{appx-region-sec}.

The algorithm for the approximate sum of a sorted list $X$ of
nonnegative $n$ numbers generates a series disjoint intervals
$R_1=[r_1, r_1'],\cdots, R_t=[r_t,r_t']$, and a series of thresholds
$b_1,\cdots, b_t$ such that each $R_i$ is an
$(1+\delta)$-approximate $b_i$-region in $X[1, r_i']$, $r_1'=n$,
$r_{i+1}'=r_i-1$, and $b_{i+1}\le {b_i\over 1+\delta}$, where
 $\delta={3\epsilon\over 4}$ and $1+\epsilon$ is
the accuracy for approximation. The sum of numbers in $X[R_i]$ is
approximated by $|R_i|b_i$. As the list $b_1>b_2>\cdots > b_t$
decreases exponentially, we can show that $t=O({1\over \epsilon}
\log n)$. The approximate sum for the input list is $\sum_{i=1}^t
|R_i|b_i$. We give a formal description of the algorithm and its
proof below.

\vskip 10pt

{\bf Algorithm Approximate-Sum($X, \epsilon,  n$)}

Input: $X[1,n]$ is a sorted list of nonnegative numbers (by
nondecreasing  order) and $n$ is the size of $X[1,n]$, and
$\epsilon$ is a parameter in $(0,1)$ for the accuracy of
approximation.

\begin{enumerate}[1.]
\item
\qquad  if $(X(n)=0)$, return $0$;

\item
\qquad let $\delta:={3\epsilon\over 4}$;

\item
\qquad let $r_1':=n$;

\item
\qquad let $s:=0$;

\item
\qquad let $i:=1$;

\item
\qquad let $b_1:={X[n]\over 1+\delta}$;

\item
\qquad while $(b_i\ge {\delta X[n]\over 3n})$ \{

\item
\qquad\qquad let $R_i:=$Approximate-Region($X, b_i, \delta, r_i'$);

\item
\qquad\qquad let $r_{i+1}':=r_{i}-1$ for $R_i=[r_i, r_i']$;

\item
\qquad\qquad let $b_{i+1}:={X[r_{i+1}']\over 1+\delta}$;

\item
 \qquad \qquad let $s_i:=|[r_i, r_i']|\cdot b_i$;

\item
 \qquad \qquad let $s:=s+s_i$;

\item
\qquad\qquad let $i:=i+1$;

\item
\qquad \};

\item
\qquad return $s$;
\end{enumerate}
{\bf End of Algorithm}

\begin{theorem}\label{main-theorem}
Let $\epsilon$ be a positive parameter. Then there is an $O({1\over
\epsilon}\min(\log n, {\log ({x_{max}\over x_{min}})})\cdot (\log
{1\over \epsilon}+\log\log n))$ time algorithm to compute
$(1+\epsilon)$-approximation for the sum of sorted list of
nonnegative numbers, where $x_{max}$ and $x_{min}$ are the largest
and the least positive elements of the input list, respectively.
\end{theorem}

\begin{proof} Assume that there are $t$ cycles executed in the while
loop of the algorithm Approximate-Sum(.). Let regions
$R_1,R_2,\cdots, R_t$ be generated.  In the first cycle of the loop,
the algorithm finds a region $R_1=[r_1, n]$ of the elements of size
at least ${X[n]\over 1+\delta}$. In the second cycle of the loop,
the algorithm finds region $R_2=[r_2, r_1-1]$ for the elements of
size at least ${X[r_1-1]\over 1+\delta}$. In the $i$-th cycle of the
loop, it finds a region $R_i=[r_i, r_{i-1}-1]$ of elements of size
at least ${X[r_{i-1}-1]\over 1+\delta}$. By the algorithm, we have
\begin{eqnarray}
j\in R_1\cup R_2\cup\cdots \cup R_t \ \ \mbox{for\ every\ $j$\ with\
} X[j]\ge {\delta X[n]\over 3n}.
\end{eqnarray}

Since each $R_i$ is an $(1+\delta)$-approximation of
${X[r_{i-1}-1]\over 1+\delta}$-region in $X[1,r_{i-1}-1]$,  $X[R_i]$
contains at least ${|R_i|\over 1+\delta}$ entries of size at least
${X[r_{i-1}-1]\over 1+\delta}$ in $X[1,r_{i-1}-1]$, $R_i$ also
contains every entry of size at least ${X[r_{i-1}-1]\over 1+\delta}$
in $X[1,r_{i-1}-1]$. Thus,
$${s_i\over
1+\delta}={|R_i|\over 1+\delta}\cdot {X[r_{i-1}-1]\over 1+\delta}\le
\sum_{j\in R_i} X[j]\le |R_i|X[r_{i-1}-1]=(1+\delta)s_i.$$
Thus,
$${s_i\over 1+\delta}\le \sum_{j\in R_i} X[j]\le (1+\delta)s_i.$$ We
have
\begin{eqnarray}
{1\over 1+\delta}\sum_{j\in R_i} X[j]\le s_i\le (1+\delta)\sum_{j\in
R_i} X[j].\label{si-app-ineqn}
\end{eqnarray}
 Thus, $s_i$ is an
$(1+\delta)$-approximation for $\sum_{j\in R_i} X[j]$. We also have
$\sum_{X[i]<{\delta X[n]\over 3n}}X[i]<{\delta X[n]\over 3}$ since
$X[1,n]$ has only $n$ numbers in total. Therefore, we have the
following inequalities:
\begin{eqnarray}
\sum_{X[i]\ge {\delta X[n]\over 3n}} X[i]&=&\sum_{i=1}^n X[i]-\sum_{X[i]< {\delta X[n]\over 3n}} X[i]\\
&\ge&\sum_{i=1}^n X[i]-{\delta\over 3}\sum_{i=1}^n X[i]\\
&=&(1-{\delta\over 3})\sum_{i=1}^n X[i]. \label{sigma-a-ineqn}
\end{eqnarray}

We have the inequalities:
\begin{eqnarray}
s&=&\sum_{i=1}^t s_i\\
&\ge& {1\over 1+\delta}\sum_{X[i]\ge {\delta X[n]\over 3n}}X[i]\ \ \ \ \ \mbox{(by\ inequality~(\ref{si-app-ineqn})))}\\
&\ge& {(1-{\delta\over 3})\over 1+\delta}\sum_{i=1}^n X[i]\ \ \ \ \ \mbox{(by\ inequality\ (\ref{sigma-a-ineqn}))}\\
 &=& {1\over {1+\delta\over 1-{\delta\over 3}}}\sum_{i=1}^n X[i]\\
 &=& {1\over 1+{{4\delta \over 3}\over 1-{\delta\over 3}}}\sum_{i=1}^n
 X[i]\\
 &\ge& {1\over 1+{4\delta \over 3}}\sum_{i=1}^n X[i]\\
&=& {1\over 1+\epsilon}\sum_{i=1}^n X[i].
\end{eqnarray}

As $R_1,R_2,\cdots$ are disjoint each other, we also have the
following inequalities:
\begin{eqnarray}
s&=&\sum_{i=1}^t s_i\\
&\le&\sum_{i=1}^t(1+\delta)\sum_{j\in R_i} X[j]\ \ \ \ \ \mbox{(by\ inequality~(\ref{si-app-ineqn}))}\\
&\le&  (1+\delta)\sum_{j=1}^n X[j]\\
&\le& (1+\epsilon)\sum_{j=1}^n X[j].
\end{eqnarray}

Therefore, the output $s$ returned by the algorithm is an
$(1+\epsilon)$-approximation for the sum $\sum_{i=1}^n X[i]$. By
Lemma~\ref{sorted-region-lemma}, each cycle in the while loop of the
algorithm takes $O((\log {1\over \delta}+\log\log n))$ time for
generating $R_i$.  For the descending chain $r_1'>r_2'>\cdots> r_t'$
with $X[r_i']\le {X[r_{i+1}']\over 1+\delta}$ and $b_i=X[r_i']\ge
{\delta X[n]\over 3n}$ for each $i$, we have that the number of
cycles $t$ is at most $O({1\over \delta}\log n)$. This is because
$X[r_t']\le {x_{max}\over (1+\delta)^t}\le {\delta X[n]\over 3n}$
for some $t=O({1\over \delta}\log n)$. Similarly, the number of
cycles $t$ is at most $O({1\over \delta}\log ({x_{max}\over
x_{min}}))$ because $X[r_t']\le {x_{max}\over (1+\delta)^t}\le
x_{min}$ for some $t=O({1\over \delta}\log ({x_{max}\over
x_{min}}))$.

Therefore, there are most $t=O({1\over \delta}\min(\log n,{{\log
{x_{max}\over x_{min}}}}))$ cycles in the while loop of the
algorithm. Therefore, the total time is $O({1\over \delta}\min(\log
n, {\log ({x_{max}\over x_{min}})})(\log {1\over \delta}+\log\log
n))=O({1\over \epsilon}\min(\log n, {\log ({x_{max}\over
x_{min}})})(\log {1\over \epsilon}+\log\log n))$. This proves
Theorem~\ref{main-theorem}.
\end{proof}

\section{Lower Bounds}\label{lower-bounds-sec}

In this section, we show several lower bounds about approximation
for the sum of sorted list. The $\Omega(\min(\log n,\log
({x_{max}\over x_{min}}))$ lower bound is based on the general
computation model for the sum problem. The lower bound $\Omega(\log
\log n))$ for finding an approximate $b$-region shows that upper
bound is optimal if using the method developed in
Section~\ref{alg-sec}. We also show that there is no sublinear time
algorithm if the input list contains one negative element.

\subsection{Lower Bound for Computing Approximate Sum}

In this section, we show a lower bound for the general computation
model, which almost matches the upper bound of our algorithm. This
indicates the algorithm in Section~\ref{alg-sec} can be improved by
at most $O(\log\log n)$ factor.

The lower bound is proved by a contradiction method.  In the proof
of the lower bound, two lists $L_1$ and $L_2$ are constructed. For
an algorithm with $o(\log n)$ queries, the two lists will have the
same answers to all queries. Thus, the approximation outputs for the
two inputs $L_1$ and $L_2$ are the same. We let the gap of the sums
from the two lists be large enough to make them impossible to share
the same constant factor approximation.

\begin{theorem}\label{app-sum-lower-bound-thm}
For every positive constant $d>1$, every $d$-approximation algorithm
for the sum of a sorted list of nonnegative numbers needs  at least
%${3\log n\over 4\log ((4+\gamma)d^2)}$
$\Omega(\min(\log n, \log {x_{max}\over x_{min}}))$ (adaptive)
queries to the list, where $\gamma$ is an arbitrary small constant
in $(0,1)$, where $x_{max}$ and $x_{min}$ are the largest and the
least positive elements of the input list, respectively..
\end{theorem}

\begin{proof}
We first set up some parameters. Let
\begin{eqnarray}
c&=&(4+\delta)d^2, \label{c-d-eqn}\\
\alpha&=&{3\over 4\log c}, \ \ \ \ \ \mbox{and}\label{alpha-d-eqn}\\
\beta&=&{3\over 4},\label{beta-d-eqn}
\end{eqnarray}
 where $\delta$
is an arbitrary small constant in $(0,1)$. Let $m$ be a positive
integer.

Let $L_0$ be a list of $t$ numbers equal to $h$ with $h\le c$ and
$t\cdot h\le \delta mc^m$, where $h,t,$ and $\delta$ will be
determined later.

Let
%list $R_m$ contain one element $c^m$, let list
%$R_{m-1}$ contain $c$ elements of $c^{m-1}$, $\cdots$, and let
list $R_{i}$ contain $c^{m-i}$ identical numbers equal to $c^{i}$
for $i=1,2,\cdots, m$. Let the first list
 $L_1'=R_1R_2\cdots R_{m}$, which is the concatenation of $R_1, R_2,\cdots, $ and $R_m$.
The list $L_1'$ has $n'=c^{m-1}+c^{m-2}+\cdots+c+1={c^{m}-1\over
c-1}$ numbers. We have $n'<c^m$ as $c>2$.
%Thus,
%\begin{eqnarray}
%\log n'<m\log c. \label{logn-logc-ineqn}
%\end{eqnarray}
Assume that an algorithm $A(.)$ only makes at most $\beta m$
%$\alpha \log n'$
 queries to output a $d$-approximation for
the sum of sorted list of nonnegative numbers.

%We have the inequalities:
%\begin{eqnarray}
%\alpha \log n'&=& {3\over 4\log c}\cdot \log n'\ \ \ \ \ \mbox{(by\
%equation\(\ref{alpha-d-eqn}))}\\
%&\le& \beta m.    \ \ \ \ \ \mbox{(by\ equation\ (\ref{beta-d-eqn})\
%and\ inequality\ (\ref{logn-logc-ineqn}))}
%\end{eqnarray}

Let $A(L_i)$ represent the computation of the algorithm $A(.)$ with
the input list $L_i$. During the computation, $A(.)$ needs to query
the numbers in the input list.
%, we will have most of the regions among $R_0,\cdots R_{m-1}$ are not queried.
Let $L_2'=R_1'R_2'\cdots R_{m}'$, where $R_i'$ has the same length
as $R_i$ and is derived from $R_i$ by the following two cases.

Let $L_i=L_0L_i'$ for $i=1,2$.

\begin{itemize}
\item
Case 1: $R_{k}$ in $L_1$ has no element queried by the algorithm
$A(L_1)$. Let $R_{k}'$ be a list of $|R_{k}|$ identical numbers
equal to that of $R_{k+1}$  (note that each element of $R_{k+1}$ is
equal to $c^{k+1}$). Since $R_{k}'$ has $c^{m-k}$ numbers equal to
$c^{k+1}$, the sum of numbers in $R_{k}'$ is $c^{m-k}\cdot
c^{k+1}=c^{m+1}$.

\item
Case 2: $R_k$ in $L_1$ has at least one element queried by the
algorithm $A(L_1)$.
 Let $R_k'=R_k$.
\end{itemize}

It is easy to verify that $L_2$ is still a nondecreasing list.
  The number of $R_i$s that
are not queried in $A(L_1)$ is at least $(m-\beta m)$, as the number
of queried elements is at most $\beta m$.

Let $S_1$ be the sum of elements in $L_1$, and $S_2$ be the sum of
elements in $L_2$. We have $S_1\le (\delta+1)mc^m$, and $S_2\ge
(m-\beta m)c^{m+1}$.
%We have ${T\over S}={({(m\over2}-o(m))(c+1)c^m\over mc^m}\ge 8d^2$.
The two lists $L_1$ and $L_2$ have the same result for running the
algorithm. Assume that the algorithm gives an approximation $s$ for
both $L_1$ and $L_2$. We have
\begin{eqnarray}
s&\le& dS_1\le d(1+\delta)mc^m\ \ \  \mbox{for\ $L_1$}, and\label{L1-ineqn}\\
{1\over d}(m-\beta m)c^{m+1}&\le& {S_2\over d}\le s \ \ \ \mbox{
for\ $L_2$}\label{L2-ineqn}.
\end{eqnarray}
 By inequalities (\ref{L1-ineqn}) and (\ref{L2-ineqn}), we have ${1\over d}(m-\beta m)c^{m+1}\le d(1+\delta)mc^m$. Thus,
${1\over d}(1-\beta )c\le d(1+\delta)$. Thus, $1-{d^2(1+\delta)\over
c}\le \beta$. By equation~(\ref{c-d-eqn}), we have
$1-{d^2(1+\delta)\over c}>1-{1\over 4}={3\over 4}=\beta$. This
brings a contradiction. Thus, the algorithm cannot give a
$d$-approximation for the sum of sorted list with at most $\beta m$
%$\alpha \log n$
 queries to the input list.

 The largest number of $L_1$ and $L_2$  is $c^m$. We can create the
 two cases for the lower bound.

 \begin{itemize}
\item
Case 1: $\log n>\log {x_{max}\over x_{min}}$.  We just let $L_0$
contains $t=n-n'$ $0$s. We have $\log {x_{max}\over x_{min}}=\log
{c^m\over c}=(m-1)\log c$. Since the algorithm has to make at least
$\beta m=\Omega(\log {x_{max}\over x_{min}})$ queries, we can see a
lower bound of $\Omega(\log {x_{max}\over x_{min}})$.

\item
Case 2: $\log n\le \log {x_{max}\over x_{min}}$. Let $L_0$ only
contain one number $h={\delta c\over n^2}$ (note $t=1$). Since the
algorithm has to make at least $\beta m=\Omega(\log n)$ queries, we
can see a lower bound of $\Omega(\log n)$.

 \end{itemize}
\end{proof}

%\begin{corollary}\label{app-sum-lower-bound-coro}
%For every positive constant $d>1$, there is no $d$-approximation
%algorithm for the sum of a sorted list of nonnegative numbers with
%at most ${3\log n\over 4\log (5d^2)}$ (adaptive) queries to the
%list.
%\end{corollary}

%??? The bound needs to be improved as we can show that $T\ge ({m\over 2}-\beta m)c^{m+1}+mc^m$.

%??? Make the lower bound as tight as possible in the constant factor level.

%In the computation model of this paper, an arithmetic operation
%between two numbers takes one unit time, which does not depend on
%the length of two numbers. The numbers in the two lists $L_1$ and
%$L_2$ are up to $c^m$. In the proof of
%Theorem~\ref{app-sum-lower-bound-thm}, we can select a fixed integer
%$k$ with $2^{k}\ge 4d^2$ for a constant $d$, and  let $c=2^{k}$. For
%each number $c^i$ for $i=1,2,\cdots, m$, it can be represented by
%the floating point expression $1.0E(ki)$, which means $c^i=2^{ki}$.
%Thus, each number involved in the proof of
%Theorem~\ref{app-sum-lower-bound-thm} has length at most $O(\log n)$
%under the floating point representation.

\subsection{Lower Bound for Computing Approximate Region}

We give an $\Omega(\log\log n)$ lower bound for the deterministic
approximation scheme for a $b$-region in a sorted input list of
nonnegative numbers. The method is that if there is an algorithm
with $o(\log\log n)$ queries, two sorted lists $L_1$ and $L_2$ of
$0,1$ numbers are constructed. They reply the same answer the each
the query from the algorithm, but their sums have large difference.
This lower bound shows that it is impossible to use the method of
Section~\ref{alg-sec}, which iteratively finds approximate regions
via a top down approach, to get a better upper bound for the
approximate sum problem.

\begin{definition}\label{1-region-def}
For a sorted list $X[1,n]$ with $0,1$ numbers by nondecreasing
order, an {\it $d$-approximate $1$-region} is a region $R=[s,n]$,
which contains the last position $n$ of $X[1,n]$, such that at least
${|R|\over d}$ numbers in $X[s,n]$ are $1$, and $X[s,n]$ contains
all the positions $j$ with $X[j]=1$, where $|R|$ is the number of
integers $i$ in $R$.
\end{definition}

\begin{theorem}\label{det-lower-bin-packing-thm}
For any parameter $d>1$, every deterministic  algorithm must make at
least $\log\log n-\log\log (d+1)$ adaptive queries to a sorted input
list for the $d$-approximate $1$-region problem.
\end{theorem}

\begin{proof}
We let each input list contain either $0$ or $1$ in each position.
Assume that $A(.)$ is a $d$-approximation algorithm for the
approximate region. Let $A(L_i)$ represent the computation of $A(.)$
with input list $L_i$. We  construct two lists $L_1$ and $L_2$ of
length $n$, and
 make sure that $A(L_1)$ and $A(L_2)$ receive the same answer
for each query to the input list. For the list of adaptive queries
generated by the algorithm $A(.)$, we generate a series of intervals

\begin{eqnarray}
[1,n]=I_0\supseteq I_1\supseteq\cdots \supseteq I_m.
\label{first-list1}
\end{eqnarray}

We also have a list
\begin{eqnarray}
[n,n]=I_0^R\subseteq I_1^R\subseteq\cdots \subseteq I_m^R,
\label{second-list1}
\end{eqnarray}
 where $m$ is the number of queries to the input list by the
 algorithm  $A(.)$ and
each $I_j^R$ is a subset of $I_j$ for $j=0,1,2,\cdots, m$.

 For each $I_j$, it is partitioned into
$I_j^L\cup I_j^R$ such that its right part $I_j^R$ is for $1$, and
its left part $I_j^L$ is undecided except its leftmost position.
Furthermore,
\begin{eqnarray}
|I_j|\ge n^{1\over 2^j} |I_j^R|, \label{stage-j-ineqn}
\end{eqnarray}
 and both $I_j$ and $R_j^R$ always
contain the position $n$, which is the final position in the input
list.

\vskip 10pt

{\bf Stage $0$}

\qquad let $I_0:=[1,n]$;

\qquad let $I_0^R:=[n,n]$;

\qquad let $L_1[1]:=L_2[1]:=0$;

\qquad let $L_1[n]:=L_2[n]:=1$;

\qquad mark every $1<i<n$ as a ``undecided" position ($1$ and $n$
are already decided);

{\bf End of Stage $0$;}

It is easy to see that inequality~(\ref{stage-j-ineqn}) holds for
Stage $j=0$.

\vskip 10pt

For an interval $[a,b]$,  $|[a,b]|$ is the number of integers in it
as defined in Definition~\ref{interval-size-def}. Assume that
$I_j=[a_j,n]$ and $I_j^R=[b_j,n]$. We assume that
inequality~(\ref{stage-j-ineqn}) holds for $j$. We also assume that
both  $L_1[i]$ and $L_2[i]$ have been decided to hold $0$ for each
$i\le a_j$;   both $L_1[i]$ and $L_2[i]$ have been decided to hold
$1$ for each $i\ge b_j$; and the other points are undecided after
stage $j$, which processes the $j$-query.

\vskip 10pt

{\bf Stage $j+1$  $(j\ge  0)$}

Assume that a position $p$ is queried  to the input list by the
$j+1$-th query ($j\ge 0$) made by the algorithm $A(.)$. We discuss
several cases.

\begin{itemize}
\item
 Case 1: $p\le a_j$. Let $I_{j+1}:=I_{j}$ and
$I_{j+1}^R:=I_{j}^R$. We have
$${|I_{j+1}|\over |I_{j+1}^R|}={|I_j|\over |I_j^R|}\ge {n^{1\over 2^{j}}}> n^{1\over 2^{j+1}}.$$
Let the answer to the $j+1$-th query be $0$ as we already assigned
$L_1[p]:=L_2[p]:=0$ in the earlier stages by the hypothesis.

\item
Case 2: $p> a_j$ and $p\in I_j^R$. Let $I_{j+1}:=I_{j}$ and
$I_{j+1}^R:=I_{j}^R$. We have
$${|I_{j+1}|\over |I_{j+1}^R|}={|I_j|\over |I_j^R|}\ge {n^{1\over 2^{j}}}> n^{1\over 2^{j+1}}.\ \ \ \mbox{(by\ the\ hypothesis)}$$
Let the answer to the $j+1$-th query be $1$ as we already assigned
$L_1[p]:=L_2[p]:=1$ in the earlier stages  by the hypothesis.

\item
Case 3: $p> a_j$ and $p\not\in I_j^R$ and ${|[p,n]|\over |I_j^R|}\ge
{\sqrt{|I_j|\over |I_j^R|}}$. Let $I_{j+1}:=[p,n]$ and
$I_{j+1}^R:=I_j^R$. We still have
$${|I_{j+1}|\over |I_{j+1}^R|}={|[p,n]|\over |I_j^R|}\ge
{\sqrt{|I_j|\over |I_j^R|}}\ge \sqrt{n^{1\over 2^{j}}}= n^{1\over
2^{j+1}}.$$

Let the answer to the $j+1$-th query be $0$, as the position $p$
will hold the number $0$. Let $L_1[i]:=L_2[i]:=0$ for each undecided
$i\le p$ (it becomes ``decided" after the assignment).

% Let $I_{j+1}=I_j$ and make $[p,b]$ as $I_{j+1}^R$.

\item
Case 4: $p> a_j$ and $p\not\in I_j^R$ and ${|[p,n]|\over |I_j^R|}<
{\sqrt{|I_j|\over |I_j^R|}}$.   Let $I_{j+1}:=I_j$ and
$I_{j+1}^R:=[p,n]$.
%Since ${|[a_j,n]|\over |I_j^R|}={|[a_j,n]|\over
%|[p,n]|}\cdot {|[p,n]|\over |I_j^R|}$,
We have the inequalities
\begin{eqnarray}
{|I_{j+1}|\over |I_{j+1}^R|}={|I_j|\over |[p,n]|}&=&{{|I_j|\over
|I_j^R|}\over {|[p,n]|\over |I_j^R|}}\\
&>&{{|I_j|\over
|I_j^R|}\over {\sqrt{|I_j|\over |I_j^R|}}}\ \mbox{(by \ the \ condition \ of \ this\ case)}\\
&=& \sqrt{|I_j|\over |I_j^R|}\ge \sqrt{n^{1\over 2^{j}}}= n^{1\over
2^{j+1}}. \ \ \ \mbox{(by\  the\ hypothesis)}
\end{eqnarray}

Let the answer to the $j+1$-th query be $1$, as the position $p$
will hold the number $1$. Let $L_1[i]:=L_2[i]:=1$ for each undecided
$i\ge p$ (it becomes ``decided" after the assignment).
\end{itemize}

{\bf End of Stage $j+1$}

\vskip 10pt

Assume that there are $m$ queries. The following final stage is
executed after processing all the $m$ queries.

\vskip 10pt

 {\bf Final Stage}

\qquad assume that $I_{m}=[a_{m},n]$ and $L_{m}^R=[b_{m}, n]$.

\qquad let $L_1[i]:=0$ for  every undecided $i<b_{m}$, and let
$L_1[i]=1$ for every undecided $i\ge b_{m}$;

\qquad let $L_2[i]:=0$ for every undecided $i\le a_{m}$, and let
$L_1[i]=1$ for every undecided $i> a_{m}$;

{\bf End of Final Stage}

\vskip 10pt

We note that the assignments to the two lists $L_1$ and $L_2$ are
consistent among all stages. In other words, if $L_i[j]$ is assigned
$a\in \{0,1\}$ at stage $k$, then $L_i[j]$ will not be assigned
$b\not=a$ at any stage $k'$ with $k<k'$, because of the two
chains~(\ref{first-list1}) and (\ref{second-list1}) in the
construction.

The two deterministic computations $A(L_1)$ and $A(L_2)$ have the
same result. We get two sorted lists $L_1$ and $L_2$ such that each
 position in $I_{m}^R$ of $L_1$ is $1$, every other position
of $L_1$ is $0$, each position in $[a_{m}+1,n]$ of $L_2$ is $1$, and
every other position of $L_2$ is $0$, where $I_{m}=[a_{m},n]$.
%All other positions of $L_1$ and $L_2$ are $0$s.

On the other hand, the numbers of $1$s of $L_1$ and $L_2$ are
greatly different. Let $D$ be the approximate $1$-region outputted
by the algorithm for the two lists. As the algorithm gives a
$d$-approximation for $L_1$, we have
\begin{eqnarray}
{|D|\over d}\le |I_{m}^R|.\label{for-L1-ineqn}
\end{eqnarray}
 As $D$ is a $d$-approximate $1$-region
for $L_2$, $D$ contains every $j$ with $X[j]=1$ (see
Definition~\ref{1-region-def}). We have
\begin{eqnarray}
|I_{m}|-1\le |D|.\label{for-L2-ineqn}
\end{eqnarray}
By inequalities (\ref{for-L1-ineqn}) and (\ref{for-L2-ineqn}),
$|I_{m}|-1\le d|I_{m}^R|$. Therefore, ${|I_{m}|-1\over |I_{m}^R|}\le
d$. Thus, ${|I_{m}|\over |I_{m}^R|}\le d+1$ as $|I_{m}^R|\ge 1$. We
have $n^{1\over 2^{m}}\le d+1$. This implies $m\ge \log\log
n-\log\log (d+1)$.
\end{proof}

\begin{corollary}\label{det-lower-bin-packing-cor} For any constant
$\epsilon\in (0,1)$, every deterministic $O(1)$-approximation
algorithm for $1$-region problem must make at least
$(1-\epsilon)\log\log n$ adaptive queries.
\end{corollary}

\subsection{Lower Bound for Sorted List with Negative Elements}

We derive a theorem that shows there is not any factor approximation
sublinear time algorithm for the sum of a list of elements that
contains both positive and negative elements.

\begin{theorem}
Let $\epsilon$ be an arbitrary positive constant. There is no
algorithm that makes at most $n-1$ queries to give
$(1+\epsilon)$-approximation for the sum of a list of $n$ sorted
elements that contains at least one negative element.
\end{theorem}

\begin{proof}
Consider a list of element $-m(m+1),  2, \cdots, 2m$. This list
contains $n=m+1$ elements. If there is an algorithm that gives
$(1+\epsilon)$-approximation, then there is an element, say $2k$,
that is not queried by the algorithm.

We construct another list that is identical to the last list except
$2k$ being replaced by $2k+1$.

The sum of the first list is zero, but the sum of the second list is
$1$. The algorithm gives the same result as the element $2k$ in the
first list and the element $2k+1$ in the second list are not queried
(all the other queries are the of the same answers). This brings a
contradiction.

Similarly, in the case that $-m(m+1)$ is not queried, we can bring a
contradiction after replacing it with $-m(m+1)+1$.
\end{proof}

\section{Implementation and Experimental Results}\label{implementation-sec}

As computing the summation of a list of elements is widely used,
testing the algorithm with program is important.
 Our algorithm has not only
theoretical guarantee for its speed and accuracy, but also
simplicity for converting into software.  We have implemented the
algorithm described in Section~\ref{alg-sec}. It has the fast
performance to compute the approximate sum of a sorted list with
nonnegative real numbers. As the algorithm is simple, it is straight
to convert it into a C++ program, which shows satisfactory
performance for both the speed and accuracy of approximation.

In the experiments conducted, we set up a loop to compute the
summation of $n=10^7$ elements. The loop is repeated $k=100$ times.
The approximation algorithm is much faster than the brute force
method to compute the approximate sum.

In order to avoid the memory limitation problem, we use an
nondecreasing  function $x(.)$, instead of a list, from integers to
double type floating point numbers. There is a function ``double
approximate\_sum(double (*x)(int), double e, int n)". If we let
function $x(i)$ return the $i$-the element of an input list, it can
also handle the input of a list of numbers, and compute its
approximate sum. In order to avoid the time consuming computation
for the square root function, we set up a table of $30$ entries to
save the values for $2^{2^k}$ with integer $k\in [-20, 9]$. This
table is enough to handle $e$ as small as $10^{-6}$ without calling
library function $sqrt(.)$ to compute the square root, and $n$ as
large as $2^{2^9}$.

When the number $n$ of numbers of the input is fixed to be $10^7$,
the speed of the software depends on the accuracy $1+e$. We let
$x(i)=i$ during the experiments. For parameter $e=0.1,0.01, 0.001$
and $0.0001$, our algorithm for the approximate sum is much faster
than the brute force method, which computes the exact sum.

Our algorithm may be slower than the brute force method when $e$ is
very small (for example $e=0.00001$).
% For example, it is slower than the brute force
%algorithm when $e=0.0001$.
 This is very reasonable from the analysis
of the algorithm as the complexity is inversely propositional to
$e$, and the algorithm Approximate-Sum(.) generates a lot of regions
$R_i$ with only one position.

\section{Conclusions and Open Problems }

We studied the approximate sum in a sorted list with nonnegative
elements. For a fixed $\epsilon$, there is a $\log\log n$ factor gap
between the upper bound of our algorithm, and our lower bound. An
interesting problem of further research is to close this gap.
%Our lower bound matches the upper bound of the
%approximation algorithm if $\epsilon$ is fixed.
%It seems there is some space to improve the current upper bound if $\epsilon$ is considered as a parameter.
%Our algorithm is generalized to handle $k$-sorted list.
%It will be interesting to close the gap of $\log\log n$ factor
%between the lower bound and upper bound for computing the
%approximate sum of a sorted list of nonnegative elements.
Another interesting problem is the computational complexity of
approximate sum in the randomized computational model, which is not
discussed in this paper.
%The algorithm presented in this paper may have some application in numerical computations.

\section{Acknowledgements}

We would like to thank Cynthia Fu for her proofreading and comments
for an earlier version of this paper.

%\bibliographystyle{abbrv}
%\bibliography{bib}

\end{document}

%% file: mac90.tex
\newcounter{savenumi}

\newtheorem{theoremfoo}{Theorem}%[section] %by chapter in report style
\newenvironment{theorem}{\pagebreak[1]\begin{theoremfoo}}{\end{theoremfoo}}

\newtheorem{propositionfoo}[theoremfoo]{Proposition}

\newtheorem{lemmafoo}[theoremfoo]{Lemma}
\newenvironment{lemma}{\pagebreak[1]\begin{lemmafoo}}{\end{lemmafoo}}
\newtheorem{conjecturefoo}[theoremfoo]{Conjecture}

\newtheorem{corollaryfoo}[theoremfoo]{Corollary}
\newenvironment{corollary}{\pagebreak[1]\begin{corollaryfoo}}{\end{corollaryfoo}}
\newtheorem{exercisefoo}{Exercise}

\newtheorem{openfoo}[theoremfoo]{Question}

\newtheorem{nttn}[theoremfoo]{Notation}

\newtheorem{dfntn}[theoremfoo]{Definition}
\newenvironment{definition}{\pagebreak[1]\begin{dfntn}\rm}{\end{dfntn}}

\newenvironment{proof}
    {\pagebreak[1]{\narrower\noindent {\bf Proof:\quad\nopagebreak}}}{\QED}
\newcommand{\floor}[1]{\left\lfloor#1\right\rfloor}
\newcommand{\ceiling}[1]{\left\lceil#1\right\rceil}
\def\nre.{$n$\/-r.e.}

\newtheorem{factfoo}[theoremfoo]{Fact}

\newcommand{\squeeze}{
\textwidth 6in
\textheight 8.8in
\oddsidemargin 0.2in
\topmargin -0.4in
}

\newtheorem{propertyfoo}[theoremfoo]{Property}

\makeatletter

\def\@makechapterhead#1{ \vspace*{50pt} { \parindent 0pt \raggedright 
 \ifnum \c@secnumdepth >\m@ne \huge\bf \@chapapp{} \thechapter. \par 
 \vskip 20pt \fi \Huge \bf #1\par 
 \nobreak \vskip 40pt } }

\def\@sect#1#2#3#4#5#6[#7]#8{\ifnum #2>\c@secnumdepth
     \def\@svsec{}\else 
     \refstepcounter{#1}\edef\@svsec{\csname the#1\endcsname.\hskip 1em }\fi
     \@tempskipa #5\relax
      \ifdim \@tempskipa>\z@ 
        \begingroup #6\relax
          \@hangfrom{\hskip #3\relax\@svsec}{\interlinepenalty \@M #8\par}
        \endgroup
       \csname #1mark\endcsname{#7}\addcontentsline
         {toc}{#1}{\ifnum #2>\c@secnumdepth \else
                      \protect\numberline{\csname the#1\endcsname}\fi
                    #7}\else
        \def\@svsechd{#6\hskip #3\@svsec #8\csname #1mark\endcsname
                      {#7}\addcontentsline
                           {toc}{#1}{\ifnum #2>\c@secnumdepth \else
                             \protect\numberline{\csname the#1\endcsname}\fi
                       #7}}\fi
     \@xsect{#5}}

\def\@begintheorem#1#2{\it \trivlist \item[\hskip \labelsep{\bf #1\ #2.}]}

\def\@opargbegintheorem#1#2#3{\it \trivlist
      \item[\hskip \labelsep{\bf #1\ #2\ (#3).}]}

\makeatother

\newif\ifshortconferences
\shortconferencesfalse
\newif\ifmediumconferences
\mediumconferencesfalse

\def\ending#1{{\count1=#1\relax
\count2=\count1
\divide\count2 by 100
\multiply\count2 by 100
\advance\count1 by -\count2
\ifnum\count1=11
th%
\else \ifnum\count1=12
th%
\else \ifnum\count1=13
th%
\else 
\count2=\count1
\divide\count1 by 10
\multiply\count1 by 10
\advance\count2 by -\count1
\ifnum\count2=1
st%
\else \ifnum\count2=2
nd%
\else \ifnum\count2=3
rd%
\else th%
\fi\fi\fi\fi\fi\fi
}}

\def\Proceedingsofthe{\ifshortconferences Proc.\else\ifmediumconferences Proc.\else Proceedings of the\fi\fi}

\newcounter{confnum}

\def\conf#1#2{%
\setcounter{confnum}{#2}%
\addtocounter{confnum}{-\csname #1zero\endcsname}%
\ifnum\value{confnum}=1%
\expandafter\ifx\csname #1One\endcsname\relax%
\Proceedingsofthe\ \arabic{confnum}\ending{\value{confnum}}\ \csname #1name\endcsname%
\else \csname #1One\endcsname\fi%
\else%
\Proceedingsofthe\
\arabic{confnum}\ending{\value{confnum}}\ \csname #1name\endcsname\fi}

\def\qsym{\vrule width0.7ex height0.9em depth0ex}

\newif\ifqed\qedtrue

\def\noqed{\global\qedfalse}

\def\qed{\ifqed{\penalty1000\unskip\nobreak\hfil\penalty50
\hskip2em\hbox{}\nobreak\hfil\qsym
\parfillskip=0pt \finalhyphendemerits=0\par\medskip}\fi\global\qedtrue}

\makeatletter
\def\eqnqed{\noqed
	\def\@tempa{equation}
	\ifx\@tempa\@currenvir\def\@eqnnum{\qsym}%
	\addtocounter{equation}{-1}\else%
    \def\@@eqncr{\let\@tempa\relax
    \ifcase\@eqcnt \def\@tempa{& & &}\or \def\@tempa{& &}%
      \else \def\@tempa{&}\fi
     \@tempa {\def\@eqnnum{{\qsym}}\@eqnnum}%  [changed theequation to @eqnnum]
     \global\@eqnswtrue\global\@eqcnt\z@\cr}\fi}

\def\eqnlabel#1#2{\if@filesw {\let\thepage\relax%
   \def\protect{\noexpand\noexpand\noexpand}%
   \edef\@tempa{\write\@auxout{\string
      \newlabel{#2}{{{#1}}{\thepage}}}}%
   \expandafter}\@tempa%
   \if@nobreak \ifvmode\nobreak\fi\fi\fi%
	\def\@tempa{equation}
	\ifx\@tempa\@currenvir\def\theequation{{#1}}% 
	\addtocounter{equation}{-1}\else%
    \def\@@eqncr{\let\@tempa\relax
    \ifcase\@eqcnt \def\@tempa{& & &}\or \def\@tempa{& &}%
      \else \def\@tempa{&}\fi
     \@tempa {\def\@eqnnum{{#1}}\@eqnnum}% [changed theequation to @eqnnum]
     \global\@eqnswtrue\global\@eqcnt\z@\cr}\fi}

\makeatother

\def\QED{\qed}

\makeatother

%% file: sortedsum.bbl
\begin{thebibliography}{10}

\bibitem{Anderson99}
I.~J. Anderson.
\newblock A distillation algorithm for floating-point summation.
\newblock {\em SIAM J. Sci. Comput.}, 20:1797–--1806, 1999.

\bibitem{Bresenham65}
J.~E. Bresenham.
\newblock Algorithm for computer control of a digital plotter.
\newblock {\em IBM Systems Journal}, 4(1):25–30, 1965.

\bibitem{CanettiEvenGoldreich95}
R.~Canetti, G.~Even, and O.~Goldreich.
\newblock Lower bounds for sampling algorithms for estimating the average.
\newblock {\em Information Processing Letters}, 53:17--25, 1995.

\bibitem{DemmelHida03}
J.~Demmel and Y.~Hida.
\newblock Accurate and efficient floating point summation.
\newblock {\em SIAM J. Sci. Comput.}, 25:1214–--1248, 2003.

\bibitem{Espelid95}
T.~O. Espelid.
\newblock On floating-point summation.
\newblock {\em SIAM Rev.}, 37:603–--607, 1995.

\bibitem{Gregory72}
J.~Gregory.
\newblock A comparision of floating point summation methods.
\newblock {\em Commun. ACM}, 15:838, 1972.

\bibitem{Har-Peled06}
S.~Har-Peled.
\newblock Coresets for discrete integration and clustering.
\newblock In {\em Proceedings of the 26th International Conference on
  Foundations of Software Technology and Theoretical Computer Science}, pages
  33--44, 2006.

\bibitem{Higham93}
N.~J. Higham.
\newblock The accuracy of floating point summation.
\newblock {\em SIAM J. Sci. Comput.}, 14:783–--799, 1993.

\bibitem{Hoefding63}
W.~Hoefding.
\newblock Probability inequalities for sums of bounded random variables.
\newblock {\em Journal of the American Statistical Association}, 58:13--30,
  1963.

\bibitem{Kahan65}
W.~Kahan.
\newblock Further remarks on reducing truncation errors.
\newblock {\em Communications of the ACM}, 8(1):40, 1965.

\bibitem{Knuth98}
D.~E. Knuth.
\newblock {\em The art of computer programming, Vol. 2: Seminumerical
  Algorithms, 3rd ed.}
\newblock Addison–Wesley, Reading, MA, 1998.

\bibitem{Linz70}
P.~Linz.
\newblock Accurate floating-point summation.
\newblock {\em Commun. ACM}, 13:361–--362, 1970.

\bibitem{Malcolm71}
M.~A. Malcolm.
\newblock On accurate floating-point summation.
\newblock {\em Commun. ACM}, 14:731–--736, 1971.

\bibitem{MotwaniPanigrahyXu07}
R.~Motwani, R.~Panigrahy, and Y.~Xu.
\newblock Estimating sum by weighted sampling.
\newblock In {\em Proceedings of the 34th International Colloquium on Automata,
  Languages and Programming}, pages 53--64, 2007.

\bibitem{Priest92}
D.~M. Priest.
\newblock {\em On Properties of Floating Point Arithmetics: Numerical Stability
  and the Cost of Accurate Computations, Ph.D. thesis}.
\newblock PhD thesis, Mathematics Department, University of California,
  Berkeley, CA, 1992.

\bibitem{ZhuYongZheng05}
Y.~K. Zhu, J.~H. Yong, and G.~Q. Zheng.
\newblock A new distillation algorithm for floating-point summation.
\newblock {\em SIAM Journal on Scientific Computing}, 26:2066--2078, 2005.

\end{thebibliography}
